\documentclass[11pt]{article}

\usepackage{times}

\makeatletter
\newcommand\subsubsubsection{\@startsection{paragraph}{4}{\z@}%
            {-2.5ex\@plus -1ex \@minus -.25ex}%
            {1.25ex \@plus .25ex}%
            {\normalfont\normalsize\bfseries}}
\makeatother
\usepackage{changepage}   
\usepackage{frcursive}
\usepackage{setspace}
\usepackage{sidecap}
\usepackage{mathrsfs}
\usepackage[multiple]{footmisc}

\usepackage{kbordermatrix,blkarray}

\renewcommand{\kbldelim}{(}
\renewcommand{\kbrdelim}{)}

\usepackage{wrapfig}
\usepackage{colortbl}
\usepackage{etex}
\usepackage{bbding}
\usepackage{dingbat}
\usepackage[ruled]{algorithm2e}
\usepackage{amsmath}
\usepackage{algorithmic}
\usepackage{amsfonts}
\usepackage{amssymb}
\usepackage{amsthm}
\usepackage{latexsym}
\usepackage{graphicx}
\usepackage{color}
\usepackage[normalem]{ulem}
\usepackage{balance}
\usepackage{textcomp}
\usepackage{nicefrac}
\usepackage{wasysym}
\usepackage{enumitem}
\usepackage{gensymb}
\usepackage{ctable}
\usepackage{url}
\usepackage{longtable}
\usepackage{citesort}

\usepackage{frcursive}

\usepackage{pifont,epsfig,rotating}

\usepackage{etoolbox}

\DeclareMathAlphabet{\mathpzc}{OT1}{pzc}{m}{it}

\usepackage{etaremune}

\usepackage{mathtools}

\definecolor{dgreyblue}{rgb}{0.26,0.3,0.46}             

\newcommand{\cA}{\mathcal{A}}

\newcommand{\cC}{\mathcal{C}}

\renewcommand{\text}[1]{\hbox{\rm \ #1\ \/}}

\newcommand{\be}[1]{\begin{equation}\label{#1}}
\newcommand{\ee}{\end{equation}}
\newcommand{\beqn}{\begin{eqnarray*}}
\newcommand{\eeqn}{\end{eqnarray*}}
\newcommand{\beq}{\begin{eqnarray}}
\newcommand{\eeq}{\end{eqnarray}}
\newcommand{\ben}{\begin{enumerate}}
\newcommand{\een}{\end{enumerate}}
\newcommand{\bi}{\begin{itemize}}
\newcommand{\ei}{\end{itemize}}
\newcommand{\eps}{\varepsilon}

\newcommand{\IE}{{\em i.e.}\xspace}
\newcommand{\tx}{^{\rm th}}

\newtheorem{claim}{Claim}

\newtheorem{theorem}{Theorem}
\newtheorem{remark}{Remark}
\newtheorem{lemma}[theorem]{Lemma}

\renewenvironment{proof}{{\noindent\bf Proof.\ }}{\hfill{\Pisymbol{pzd}{113}}\vspace{0.1in}}
\newenvironment{proof-sketch}{{\noindent\bf Sketch of Proof.\ }}{\hfill{\Pisymbol{pzd}{113}}\vspace{0.1in}}

\newcommand{\NP}{\mathsf{NP}}

\newcommand{\cS}{\mathcal{S}}

\newcommand{\cP}{\mathcal{P}}
\newcommand{\EA}{{\em et al.}\xspace}
\newcommand{\TB}{\vspace{-0.1ex}}\newcommand{\TiE}{\setlength{\itemsep}{-1ex}}

\newcommand{\comment}[1]{}

\newcommand{\EG}{{\it e.g.}\xspace}
\newcommand{\FI}[1]{Fig.~\ref{#1}\xspace}

\newcommand{\mispc}{\textsc{Mis}$_{\mathrm{PC}}$}

\definecolor{columbiablue}{rgb}{0.61, 0.87, 1.0}

\allowdisplaybreaks

\newcommand{\aalpha}{{\mathsf{PartyA}}}
\newcommand{\bbeta}{{\mathsf{PartyB}}}
\newcommand{\eeta}{{\mathsf{Pop}}}
\newcommand{\effgap}{{\mathsf{Effgap}}}
\newcommand{\bias}{{\mathsf{Bias}}}
\newcommand{\partya}{Party~\textbf{A}}
\newcommand{\partyb}{Party~\textbf{B}}
\newcommand{\nsc}{{\mbox{\sf N-Seat-C}}}
\newcommand{\nvc}{{\mbox{\sf N-Vote-C}}}
\newcommand{\nsm}{{\mbox{\sf N-Seat-M}}}
\newcommand{\nvm}{{\mbox{\sf N-Vote-M}}}
\newcommand{\wv}{{\mbox{\sf Wasted-Votes}}}
\newcommand{\copact}{\mathscr{C}}

\title{On partisan bias in redistricting: computational complexity meets the science of gerrymandering}

\author{Tanima Chatterjee \hspace*{0.3in} Bhaskar DasGupta
\\
Department of Computer Science 
\\
University of Illinois at Chicago
\\
Chicago, IL 60607, USA
\\
{\sf tchatt2@uic.edu, bdasgup@uic.edu}
}

\date{(Preliminary version)}

\begin{document}

\maketitle

\begin{abstract}
\thispagestyle{empty}
The main topic of this paper is ``gerrymandering'', namely 
the curse of \emph{deliberate} creations of district maps with \emph{highly} asymmetric electoral outcomes
to disenfranchise voters, and it has a \emph{long} legal history going back as early as $1812$.
Measuring and eliminating gerrymandering has enormous far-reaching implications to 
sustain the backbone of democratic principles of a country or society.

Although there is no dearth of legal briefs filed in courts involving many aspects of 
gerrymandering over many years in the past, 
it is only more recently that mathematicians and \emph{applied} computational researchers have started to 
investigate this topic. 
However,  
it has received relatively \emph{little} attention so far from the \emph{computational complexity researchers} (where
by ``computational complexity researchers'' we mean researchers dealing with theoretical analysis of 
computational complexity issues of these problems, such as polynomial-time solvabilities, 
approximability issues, \emph{etc}.).
There could be \emph{several} reasons for this, such as descriptions of these problem 
non-CS non-math (often legal or political) journals that are 
\emph{not} very easy for theoretical CS (\emph{TCS}) people to follow, or 
the lack of effective collaboration between \emph{TCS} researchers and other (perhaps non-CS) 
researchers that work on these problems accentuated by the lack of coverage of these topics 
in \emph{TCS} publication venues. 
One of our modest goals in writing this article is to improve upon this situation
by stimulating \emph{further} interactions between the science of gerrymandering 
and the \emph{TCS} researchers. To this effect, our main contributions in this article are \emph{twofold}: 
\begin{enumerate}[label=$\triangleright$]
\item
We provide formalization of several models, related concepts, and corresponding problem statements using \emph{TCS} frameworks 
from the descriptions of these problems as available in existing non-CS-theory (perhaps legal) venues.
\item
We also provide computational complexity analysis of some versions of these problems, leaving 
other versions for future research.
\end{enumerate}
The goal of writing article is \emph{not} to have the final word on gerrymandering, but to introduce a series of
concepts, models and problems to the \emph{TCS} community and to show that science of gerrymandering 
involves an \emph{intriguing} set of partitioning problems involving geometric and combinatorial
optimization.
\end{abstract}


\vfill

\noindent
\textbf{\small Disclaimer}: 
{\small
The authors were \emph{not} supported, financially or otherwise,
by any political party. The research results reported in this paper are purely \emph{scientific} and reported as they are
\emph{without} any regard to which political party they may be of help (if at all).
}

\clearpage

\section{Introduction}
\label{sec-intro}
\setcounter{page}{1} 

\begin{wrapfigure}{r}{3.5cm}
\vspace*{-0.2in}
\includegraphics[scale=0.6]{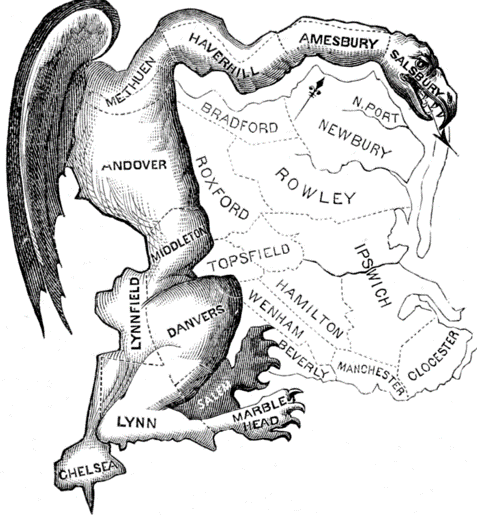}
\vspace*{-0.4in}
\caption{\label{fig1}\cite{wiki}``Gerry'' and ``salamander'' districts, $1812$ state senate election, Massachusetts.}
\end{wrapfigure}
Gerrymandering, namely \emph{deliberate} creations of district maps with \emph{highly} asymmetric electoral outcomes
to disenfranchise voters, has continued to be a curse to fairness of electoral systems in USA 
for a long time in spite of general public disdain for it.
There is a \emph{long} history of this type of voter disenfranchisement going back as early as $1812$ 
when the specific term ``gerrymandering'' was coined after a redistricting of the senate election map of the state of 
Massachusetts resulted in a South Essex district taking a shape that resembled a \emph{salamander} (see \FI{fig1}).
There is an elaborate history of \emph{litigations} involving gerrymandering as well.
In $1986$ the US Supreme Court (\emph{SCOTUS}) ruled 
that gerrymandering is \emph{justiciable}~\cite{t1}, 
but they could \emph{not} agree on an effective way of estimating it. 
In $2006$, \emph{SCOTUS} opined
that a measure of \emph{partisan symmetry} may be a helpful
tool to understand and remedy gerrymandering~\cite{t2}, but again a precise quantification of partisan 
symmetry that will be acceptable to the courts was left \emph{undecided}.
Indeed, formulating precise and computationally efficient measures for 
partisan bias (\IE, lack of partisan symmetry) 
that will be acceptable in \emph{courts} may be considered \emph{critical} to removal of 
gerrymandering\footnote{Even though measuring partisan bias \emph{is} a non-trivial issue, it has nonetheless been observed that 
two frequent indicators for partisan bias are 
\textbf{cracking}~\cite{PLB11} (dividing supporters of a specific party between two
or more districts when they could be a majority in a single district) 
and \textbf{packing}~\cite{PLB11} (filling a district with more supporters of a specific party 
as long as this does not make this specific party the winner in that district).
Other partisan bias indicators include 
\textbf{hijacking}~\cite{PLB11} (re-districting to force two incumbents to run against each other in one district)
and \textbf{kidnapping}~\cite{PLB11} (moving an incumbent's home address into another district).}\footnote{See Section~\ref{sec-scotus}
regarding the impact of the \emph{SCOTUS} gerrymandering ruling on 06/27/2019 on future 
gerrymandering studies.}.

Although there is no dearth of legal briefs filed in courts involving gerrymandering over many years in the past, 
it is only more recently that mathematicians and \emph{applied} computational researchers have started to 
investigate this topic, perhaps due to the tremendous progress in high-speed computation in the last two decades.
For example, researchers in~\cite{BK87,T73,TL67,NGCH90,C85,ND78,J94,Alt02,AG94}
have made conceptual or empirical attempts at 
quantifying gerrymandering and devising redistricting methods to optimize such 
quantifications using well-known notions such as 
\emph{compactness} and \emph{symmetry}, whereas 
researchers in~\cite{CDPAS18,Alt02,wcl16,LWCW16,TL67,CR15} have investigated 
designing efficient \emph{heuristic approach} and other computer simulation approaches for this purpose.
Two recent research directions deserve specific mentions here. 
In the first direction, researchers Stephanopoulos and McGhee
in several papers such as~\cite{m14,sm15} introduced a new gerrymandering measure called the 
\emph{efficiency gap} that attempts to minimize the 
absolute difference of total wasted
votes between the parties in a two-party electoral system, and 
very importantly, at least from a legal point of view, this measure was found legally convincing in a US appeals court in 
a case that claims that the legislative map of the state of Wisconsin is gerrymandered.
In another direction, and perhaps of considerable interest to the \emph{algorithmic game theory} researchers,
the authors in a recent paper~\cite{PPY17} formulated the redistricting process as a two-person game and 
analyzed the performances of two kinds of protocols for such games.

\subsection{Why write this article and why theoretical computer science researchers should care?}

Somewhat unfortunately, even though the science of gerrymandering have received 
varying degrees of attention from legal researchers, mathematicians and applied computational researchers,  
it has received relatively little attention so far from the theoretical computer science (\emph{TCS}) researchers (where by 
``\emph{TCS} researchers'' we mean researchers dealing with theoretical analysis of 
computational complexity issues of these problems, such as polynomial-time solvabilities, 
fixed-parameter tractabilities, approximability issues, \emph{etc}.), \emph{except} few recent results such as~\cite{CDPAS18}.
In our opinion there are several reasons for this. 
Often, some of these problems are described in ``non-CS non-math'' journals in a way that may \emph{not} be very precise 
and may \emph{not} be very easy for \emph{TCS} researchers to follow. Another possible reason is 
the \emph{lack} of effective collaboration between \emph{TCS} researchers and other (perhaps non-CS) 
researchers working on these problems, perhaps accentuated by the lack of coverage of these topics 
in \emph{TCS} publication venues. One of our goals in writing this article is to improve upon this situation. 
To this effect, the article is motivated by the following two high-level aims: 
\begin{description}[leftmargin=0.5cm]
\item[(I) Formalization of models and problem statements:]
Our formal definitions and descriptions need to satisfy two (perhaps mutually conflicting) goals. 
The levels of abstraction should be as close to their real-world applications as possible but 
should still make the problems sufficiently interesting so as to 
to attract the attention of the \emph{TCS} researchers.
\item[(II) Computational complexity analysis:]
We provide computational complexity analysis of some versions of these problems, leaving 
other versions for future research.
\end{description}
%
Task~\textbf{(I)} may \emph{not} necessarily be as straightforward as it seems, especially since 
descriptions of some of the problem variations may come from non-CS-theory (perhaps legal) venues.
Regarding Task~\textbf{(II)}, one may wonder why computational complexity analysis (including 
computational hardness results) may of be practical interest at all. To this, we point out a few reasons.
\begin{enumerate}[label=$\triangleright$,leftmargin=*]
\item
When a particular type of gerrymandering solution \emph{is} found acceptable in courts, one would eventually need to
develop and implement a software for this solution, especially for large US states such as California and Texas
where manual calculations may take too long or may not provide the best result. Any exact or approximation algorithms designed by 
\emph{TCS} researchers would be a valuable asset in that respect. Conversely, appropriate computational
hardness results can be used to convince a court to \emph{not} apply that measure for specific US states due to 
practical infeasibility.
\item
Beyond scientific implications, 
\emph{TCS} research works may also be expected to have a beneficial impact on the US judicial system.
Some justices, whether at the Supreme Court level or in lower courts, 
seem to have a reluctance to taking mathematics, statistics and computing seriously~\cite{R03,Fa89}.
\emph{TCS} research may be able to help showing that the theoretical methods, whether complicated or not 
(depending on one's background), \emph{can} in fact yield fast accurate
computational methods that can be applied to ``un-gerrymander'' the currently gerrymandered maps. 
\end{enumerate}

\subsection{Remarks on the impact of the \emph{SCOTUS} gerrymandering ruling}
\label{sec-scotus}

As this article was being written, 
\emph{SCOTUS} issued a ruling on 06/27/2019 on 
two gerrymandering cases~\cite{RuCo19}.
However, the ruling does not eliminate the need for future gerrymandering studies.
While \emph{SCOTUS} agreed that gerrymandering was anti-democratic, it decided that 
it is best settled at the legislative and political level, and it encouraged solving the problem 
at the state court level and delegating legislative redistricting to 
independent commissions via referendums. Both of the last two remedies do require further 
scientific studies on gerrymandering. It is also possible that a future \emph{SCOTUS}
may overturn this recent ruling. 

\section{Precise formulations of several gerrymandering problems}

We assume for the rest of the paper that our political system consists of two parties \emph{only},  
namely \partya\ and \partyb.
This means that we ignore negligible third-party votes as is commonly done by researchers interested in
two-party systems. Although some of our concepts can be extended for three or more 
major parties, we urge caution since gerrymandering for multi-party systems may need different definitions.

\begin{figure}[h]
\centerline{\includegraphics[scale=0.9]{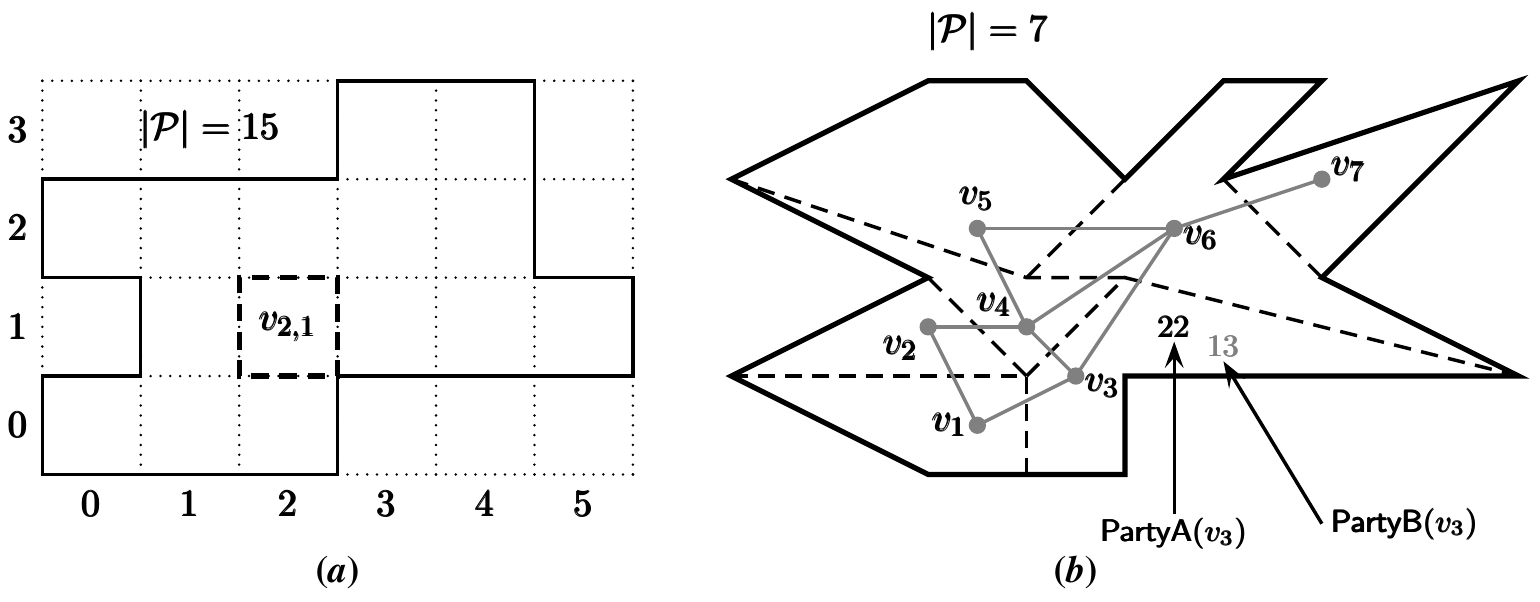}}
\caption{\label{ex1-fig}(\emph{a}) A rectilinear polygon map $\cP$ of size $15$ placed on a grid of size $6\times 4$; the cell 
$v_{2,1}$ is shown.
(\emph{b}) An arbitrary polygon map $\cP$ of size $7$. The corresponding planar graph is shown in gray.}
\end{figure}

\subsection{Input data and its granularity levels}

The topological part of an input is generically referred to a ``map'' $\cP$ which is partitioned into 
\emph{atomic elements} 
or \emph{cells} 
(\EG, subdivisions of counties or voting tabulation districts in legal gerrymandering literatures). 
The following two types of maps may be considered.
\begin{description}[leftmargin=0.5cm]
\item[Rectilinear polygon $\pmb{\cP}$ without holes (\FI{ex1-fig}$\!$(\emph{a})):]
For this case, $\cP$ is placed on a unit grid of size $m\times n$.
Then, the atomic elements (cells) of $\cP$ are identified with individual unit squares 
of the grid inside $\cP$.
We will refer to the cell
on the $i\tx$ row and $j\tx$ column 
by $v_{i,j}$ for $0\leq i<m$ and $0\leq j<n$. 
\item[Arbitrary polygon $\pmb{\cP}$ without holes (\FI{ex1-fig}$\!$(\emph{b})):]
For this case, $\cP$ is an \emph{arbitrary}
simple polygon, and 
the atomic elements (cells) of $\cP$ are arbitrary sub-polygons (without holes) inside $\cP$.
Such a map can also be thought of a planar graph $G(\cP)$ whose nodes 
are the cells, and an edge connects two cells if they share a portion of the boundary of non-zero measure.
Note that although the planar graph for a \emph{given} polygonal map is unique, 
for a given planar graph there are many polygonal maps. 
\end{description}
In either case, the \emph{size} $|\cP|$ of the map is the number of cells (\emph{resp}., nodes) in it and, 
for a cell (\emph{resp}., a node) 
$y$ and a sub-polygon $\cP'$ inside the polygonal map $\cP$ (\emph{resp}., a sub-graph $G'$ of $G(\cP)$)
the notation $y\in\cP'$ will indicate that $y$ is inside $\cP'$ (\emph{resp}., $y$ is a node of $G'$). 
Every cell or node $y$ of a map has the following numbers associated with it
(see~\FI{ex1-fig}$\!$(\emph{b})):
\begin{enumerate}[label=$\blacktriangleright$]
\item
A strictly positive integer $\eeta(y)>0$ indicating the ``total population'' inside $y$.
\item
Two non-negative integers $\aalpha(y),\bbeta(y)\geq 0$ 
such that $\aalpha(y)+\bbeta(y)=\eeta(y)$. 
$\aalpha(y)$ and $\bbeta(y)$ denotes 
the total number of voters for \partya\ and \partyb, respectively.
\end{enumerate}
In addition to the above numbers, we are also given a positive integer $1<\kappa<|\cP|$ 
that denotes the \emph{required} (legally mandated) number of districts\footnote{\textbf{This is a hard 
constraint} since a map with a different value of $\kappa$ would be \emph{illegal}. This precludes one from
designing an approximation algorithm in which the value of $\kappa$ changes even by just $\pm 1$, and conversely 
a computational hardness result for a value of $\kappa$ does \emph{not} necessarily imply a 
similar result for another value of $\kappa$.}. 
Based on existing literatures, three types of granularities of these numbers in the input data can be formalized: 
\begin{description}[leftmargin=0.5cm]
\item[Course granularity:]
For this case, the $\eeta(y)$'s are numbers of arbitrary size, and thus the total number of bits needed 
to represent the $\eeta(y)$'s
(\IE, $\sum_y \lceil \log_2 (1+\eeta(y))\rceil$)
contributes to the size of the input. 
This kind of data is obtained, for example, when one uses 
data at the ``county'' level~\cite{CDPAS18} or ``census block \emph{group}'' level~\cite{CDO00,D15}.
\item[Fine granularity:]
For this case, for every cell or node $y$ we have $0<\eeta(y)\leq c$ for some \emph{fixed constant} $c>0$.
This kind of data is obtained, for example, when one uses data  at
the ``Voting Tabulation District'' (VTD) level\footnote{VTDs are often the \emph{smallest} units in a US state for 
which the election data are available.} or at the ``census block'' level.
\item[Ultra-fine granularity:]
For this case, $\eeta(y)=c$ for some \emph{fixed constant} $c>0$ for every cell or node $y$.
If the different 
$\eeta(y)$'s in the fine granularity case do not differ from each other too much then 
depending on the optimization objective it may be possible to approximate the fine granularity
by an ultra-fine granularity.
\end{description}

\subsection{Legal requirements for valid re-districting plans}

Let $\cS$ 
denote the set of all cells (\emph{resp}., all nodes) 
in the given polygonal map $\cP$ (\emph{resp}., the planar graph $G(\cP)$).
A \emph{districting scheme} is a partition of $\cS$ 
into $\kappa$ subsets of cells (\emph{resp}., nodes), say $\cS_1,\dots,\cS_\kappa$. 
One absolutely legally required condition is the following: 

\begin{quote}
``every $\cS_j$ 
must be a connected polygon\footnote{For our purpose, two polygon sharing a single
point is assumed to be disconnected from each other.} 
(\emph{resp}., a connected subgraph)''.
\end{quote}

\noindent
For convenience, we define the following quantities for each $\cS_j$:
%
\begin{description}[leftmargin=0.5cm]
\item[Party affiliations in $\cS_j$:] 
$\aalpha(\cS_j)=\sum_{y\in\cS_j}\aalpha(y)$ 
and 
$\bbeta(\cS_j)=\sum_{y\in\cS_j}\bbeta(y)$.
\item[Population of $\cS_j$:] 
$\eeta(\cS_j)=\aalpha(\cS_j)+\bbeta(\cS_j)$.
\end{description}
Then, another legally mandated condition in its two forms can be stated as follows. 
\begin{description}[leftmargin=0.5cm]
\item[Strict partitioning criteria:]
Ideally, one would like 
$\{\cS_1,\dots,\cS_\kappa\}$ 
to be a 
(\emph{exact}) $\kappa$-\emph{equipartition} of $\cS$, 
\IE, 
\[
\forall\, j:\, \eeta(\cS_j)\in\{\lfloor\eeta(\cS)/\kappa\rfloor,\, \lceil\eeta(\cS)/\kappa\rceil\}
\]
\item[Approximately strict partitioning criteria:]
In practice, it is \emph{nearly} impossible to satisfy the strict partitioning criteria. 
To alleviate this difficulty, the exactness of equipartition is \emph{relaxed} by 
allowing 
$\eeta(\cS_1)$, $\dots$, $\eeta(\cS_\kappa)$
to differ from each other within an acceptable range.
To this effect, we define an   
$\eps$-\emph{approximate} $\kappa$-equipartition of $\cS$ 
for a given $\eps>0$ 
to be one that satisfies 
$
\dfrac{
\max_{1\leq j\leq\kappa} \left\{ \eeta(\cS_j) \right\}
}
{
\min_{1\leq j\leq\kappa} \left\{ \eeta(\cS_j) \right\}
}
\leq 1+\eps 
$.
Rulings such as~\cite{noteq} seem to suggest that the courts may allow a maximum value of 
$\eps$ in the range of $0.05$ to $0.1$. Another possibility is to have an \emph{additive} 
$\delta$-approximation to the strict partitioning criterion by allowing 
$
\max_{1\leq j\leq\kappa} \left\{ \eeta(\cS_j) \right\} \leq 
\min_{1\leq j\leq\kappa} \left\{ \eeta(\cS_j) \right\}
\,+\,\delta 
$.
\end{description}

\subsection{Optimization objectives to eliminate partisan bias}

We describe a few objective functions for optimization to remove partisan bias 
(in \emph{TCS} frameworks) that have been proposed
in existing literatures or court documents\footnote{We remind the reader that there is \emph{no} one single objective function 
that has been universally accepted in all or most court cases, and it is likely that new objectives 
will be proposed in the coming years.}.
Let 
$\cS_1,\dots,\cS_\kappa$ 
be the set of $\kappa$ districts (partitions) 
of the set of all cells (\emph{resp}., nodes) $\cS$ 
in the given polygonal (\emph{resp}., planar graph) map.
We first define a few related useful notations and concepts.  
\begin{description}[leftmargin=0.5cm]
\item[Winner of a district $\cS_j$:]
Clearly if 
$\aalpha(\cS_j)>\eeta(\cS_j)/2$
then \partya\ should be the winner and if 
$\bbeta(\cS_j)>\eeta(\cS_j)/2$
then \partyb\ should be the winner. What if 
$\aalpha(\cS_j)=\bbeta(\cS_j)=\eeta(\cS_j)/2$ ?
Most existing research works assigned the district to a \emph{specific preferred} party (\EG, \partya)
always for this case, so we will assume this \emph{by default}. However, in reality, a (fair) \emph{coin-toss}
is often used to decide the outcome\footnote{\emph{Please do not underestimate the power of a coin toss}.
The $2017$ election for the $94\tx$ district for house of 
delegates in the state of Virginia was decided by a coin toss, and in fact this also decided the 
legislative control of one of the chambers of the state.}.
\item[Normalized seat counts and seat margins of the two parties:]
\begin{gather*}
\begin{array}{cc}
\textstyle
\nsc(\mbox{\partya})
=
{
\big| \left\{ 
\cS_j \,:\,
\mbox{\partya\ wins $\cS_j$}
\right\} \big|
}
/
{\kappa},
&
\textstyle
\,\,
\nsm(\mbox{\partya})
=
\nsc(\mbox{\partya})
\,-\,
\nicefrac{1}{2}
\\
[3pt]
\textstyle
\nsc(\mbox{\partyb})
=
1 \,-\, \nsc(\mbox{\partya}),
&
\textstyle
\,\,
\nsm(\mbox{\partyb})
=
\nsc(\mbox{\partyb})
\,-\,
\nicefrac{1}{2}
\end{array}
\end{gather*}
\item[Normalized vote counts and vote margins of the two parties:]
%
\begin{gather*}
\begin{array}{cc}
\textstyle
\nvc(\mbox{\partya}) = {\aalpha(\cS)}/{\eeta(\cS)}, & 
\textstyle
\,\,\,\,
\nvm(\mbox{\partya}) = \nvc(\mbox{\partya}) - \nicefrac{1}{2}, 
\\
[3pt]
\nvc(\mbox{\partyb}) = {\bbeta(\cS)}/{\eeta(\cS)},
&
\textstyle
\,\,\,\,
\nvm(\mbox{\partyb}) = \nvc(\mbox{\partyb}) - \nicefrac{1}{2}
\end{array}
\end{gather*}
\item[Wasted votes:]
For a district $\cS_j$, the wasted votes (\IE, the votes whose absence would not have altered the election)
for the two parties are defined as follows~\cite{sm15,m14}: 
\begin{gather*}
\wv(\cS_j,\mbox{\partya}) =
\left\{
\begin{array}{r l}
\aalpha(\cS_j) - (\eeta(\cS_j)/2), 
      & \mbox{if \partya\ is the winner of $\cS_j$
			}
\\
[3pt]
\aalpha(\cS_j), & \mbox{otherwise}
\end{array}
\right.
\\
\wv(\cS_j,\mbox{\partyb}) =
\left\{
\begin{array}{r l}
\bbeta(\cS_j) - (\eeta(\cS_j)/2), 
      & \mbox{if \partyb\ is the winner of $\cS_j$}
\\
[3pt]
\bbeta(\cS_j), & \mbox{otherwise}
\end{array}
\right.
\end{gather*}
%
\end{description}
Without loss of generality, assume that 
$\aalpha(\cS)\geq\bbeta(\cS)$.
Based on the above notions, we can now describe a few optimization objectives:
\begin{description}[leftmargin=0.5cm]
\item[Seat-vote equation:]
For the \emph{decision version} of this problem, we are required to 
produce a re-districting plan that exactly satisfies 
a relationship between between normalized seat counts and normalized vote counts 
between the two parties.
The relationship was stated by~\cite{T73} as
\begin{gather}
\nsc(\mbox{\partya})
\,/\,
\nsc(\mbox{\partyb})
\approx
\big(
\aalpha(\cS)
\,/\,
\bbeta(\cS)
\big)^\rho
\label{eq1}
\end{gather}
where $\rho$ is a positive number and $\approx$ denotes almost equality.
Kendall and Stuart in~\cite{KS50}
argued in favor of $\rho=3$ using some stochastic models. 
Some special cases of Equation~\eqref{eq1} are as follows: 
\begin{quote}
\textbf{Proportional representation:} 
$\rho=1$, 
\hspace*{0.5in}
\textbf{Winner-take-all:}
$\rho=\infty$.
\end{quote}
In practice, a value of $\rho\in [1,3]$ is considered to be a reasonable choice. 
For an \emph{optimization version} of this problem,
\textbf{assuming $\pmb{\nsc(\mbox{\partyb})>0}$ and assuming \partya\ has the responsibility to 
do the re-districting}\footnote{In other words, \partya\ chooses the districts 
in an attempt to his/her desirable value for $\nsc\left(\mbox{\partya}\right)$.}, 
we define an (asymptotic) $\eps$-approximation ($\eps\geq 1$) as a 
solution that satisfies 
\begin{gather}
\textstyle
\eps^{-1}
\lim\limits_{\eeta(\cS)\to\infty}
\left(
\frac{
\aalpha(\cS)
}
{
\bbeta(\cS)
}
\right)^\rho
\leq
\textstyle
\lim\limits_{\kappa\to\infty}
\left(
\frac{
\nsc\left(\mbox{\partya}\right)
}
{
\nsc\left(\mbox{\partyb}\right)
}
\right)
\leq 
\eps
\lim\limits_{\eeta(\cS)\to\infty}
\left(
\frac{
\aalpha(\cS)
}
{
\bbeta(\cS)
}
\right)^\rho
\label{eq2}
\end{gather}
Equation~\eqref{eq2} is obviously ill-defined when 
$\nsc(\mbox{\partyb})=0$, which may indeed happen in practice for smaller values of $\kappa$ such as $\kappa=2$.
We introduce appropriate modifications to Equation~\eqref{eq2} to avoid this in the following manner.
If $\nsc(\mbox{\partyb})=0$  then 
$\nsc(\mbox{\partya})/\kappa=1$ and thus an exact version of the seat-vote equation would intuitively want 
$\aalpha(\cS)/\eeta(\cS)=1$ no matter what $\rho$ is.
Thus, when $\nsc(\mbox{\partyb})=0$,
we consider such a solution
as an $\eps$-approximation where 
\begin{gather}
\textstyle
\eps=
\lim_{\eeta(\cS)\to\infty}
\left(\nicefrac{\aalpha(\cS)}{\eeta(\cS)}\right)^{-1}
\tag*{\eqref{eq2}$'$}
\end{gather}
\item[Efficiency gap:]
The goal here is to \emph{minimize} the \emph{absolute difference} of \emph{total} wasted
votes between the parties, \IE, we need to find a partition 
that minimizes
\[
\textstyle
\effgap_{\kappa}(\cS,\cS_1,\dots,\cS_\kappa)
=
\left| \,
\sum_{j=1}^\kappa 
\big( \, \wv(\cS_j,\mbox{\partya}) \,-\, \wv(\cS_j,\mbox{\partyb}) \, \big)
\,\right|
\]
\item[Partisan bias:]
Partisan bias is a \emph{deviation} from bipartisan symmetry that favors one party over the other. 
The underlying assumption in using this very popular measure 
is that both the parties should expect to receive
the same number of seats given the same vote proportion, \IE, for example, 
if 
$\nvc(\mbox{\partya}) = 0.7$
and the redistricting plan results in 
$\nsc(\mbox{\partya}) = 0.4$
then assuming 
$\nvc(\mbox{\partya}) = 1-0.7=0.3$
the same redistricting plan should result in 
$\nsc(\mbox{\partya}) = 1-0.4=0.6$.
However, since the precise distribution of voters when 
$\nvc(\mbox{\partya}) = 0.3$ is not known, the distribution is generated artificially 
possibly based on some assumptions (which may not always be acceptable to court). 
Mathematically, a measure of partisan bias can be computed in the following manner.
\begin{enumerate}
\item
Let $\alpha=\nvc(\mbox{\partya})-\nvc(\mbox{\partyb})$. Note that $\alpha\in [0,1]$.
\item
Select $\beta_1,\dots,\beta_\kappa\in [0,1]$ such that 
$\beta_1+\dots+\beta_\kappa=\alpha$. 
These choices depend upon the \emph{population shift model} being used.
\item
For every district $\cS_j$, we create a district $\widetilde{\cS_j}$
that corresponds to the same region (sub-polygon or sub-graph) but with the following parameters changes:
\begin{gather*}
\aalpha(\widetilde{\cS_j}) = \aalpha(\cS_j) \,-\, \beta_j\eeta(\cS), 
\,\,\,\,
\bbeta(\widetilde{\cS_j}) = \bbeta(\cS_j) \,+\, \beta_j\eeta(\cS)
\end{gather*}
Note that 
$\widetilde{\cS_1},\dots,\widetilde{\cS_\kappa}$ is another legally valid re-districting plan for $\cS$
but for this new plan the normalized vote count for \partya\ is given by
\begin{gather*}
\textstyle
\frac{ \sum_{j=1}^\kappa \aalpha(\widetilde{\cS_j}) } {\eeta(\cS)}
=
\frac{ \sum_{j=1}^\kappa \big( \, \aalpha(\cS_j) \,-\, \beta_j\eeta(\cS) \,\big) } {\eeta(\cS)}
=
\frac{ \aalpha(\cS) \,-\, \alpha\eeta(\cS) \,\big) } {\eeta(\cS)}
=
\nvc(\mbox{\partyb})
\end{gather*}
\item
Recalculate the normalized seat count 
$\widetilde{\nsc}(\mbox{\partya})$
for \partya\ for this new partition
$\widetilde{\cS_1},\dots,\widetilde{\cS_\kappa}$.
\item
Define the measure of bias as 
$
\bias_{\kappa}(\cS,\cS_1,\dots,\cS_\kappa)
=
\big|
\widetilde{\nsc}(\mbox{\partya})
\,-\,
\nsc(\mbox{\partya})
\big|
$.
\end{enumerate}
The goal is then to find a partition 
$\cS_1,\dots,\cS_\kappa$ to \emph{minimize} 
$\bias_{\kappa}(\cS,\cS_1,\dots,\cS_\kappa)$.
\item[Geometric compactness of a polygonal district $\cS_j$:]
The primary goal of using this measure is to ensure that polygonal districts do not have 
``unusually weird'' shapes (\emph{cf}.\ \FI{fig1}). 
A most commonly used compactness measure is the so-called
``Polsby-Popper compactness measure''~\cite{PP91} given by 
$
\copact(\cS_j)=c\,A/B^2
$
where $A$ is the area and $B$ is the length of the perimeter of $\cS_j$, 
and $c>0$ is a suitable constant ($c=4\pi$ was used in~\cite{O78}).
The computational problem is then to find a re-districting plan 
such that 
$L_1\leq\copact(\cS_j)\leq L_2$
for all $j$
for two given bounds $L_1$ and $L_2$. 
\end{description}
In addition to what is discussed above, there are other constraints and optimization criteria, such as 
\emph{responsiveness} (also called \emph{swing ratio}), 
\emph{equal vote weight} and \emph{declination}, 
that we did not discuss; the reader is referred to references such as~\cite{W18,MB15,amici04} 
for informal discussions on them.

\subsection{Prior relevant computational complexity research}

To our knowledge, 
the most relevant prior non-trivial computational complexity (\IE, approximation hardness, approximation algorithms, \emph{etc}.) 
article regarding gerrymandering is~\cite{CDPAS18}.
The article~\cite{CDPAS18} exclusively dealt with the efficiency gap measure, and provided some non-trivial 
approximation hardness and approximation algorithms in addition to designing and implementing a practical
algorithm for this case which works well on real maps.
In the terminologies of this article, \cite{CDPAS18} showed that 
minimization of the efficiency gap measure for rectilinear polygonal maps with coarse grain inputs 
and strict partitioning criteria
does \emph{not} admit \emph{any} non-trivial polynomial-time approximation in the worst case, but \emph{does} admit
polynomial-time approximation algorithms when \emph{further} constraints are added to the problem.
In addition, \cite{CDPAS18} and \cite[p. 853]{sm15}
also observed that 
$\mathsf{Effgap}_{\kappa}(\cS,\cS_1,\dots,\cS_\kappa)/\eeta(\cS)
= | \, 2 \,\times\, \nvm(\mbox{\partya}) \,-\, \nsm(\mbox{\partya}) \,|$.

\section{Our computational complexity results}

\textbf{Before stating our technical results, we remind the reader about the following obvious but important 
observations}. Consider the following combinations for a pair $(X,Y)$:
\begin{enumerate}[label=$\triangleright$]
\item
$X$ is rectilinear polygonal input and $Y$ is arbitrary polygonal input (equivalently, a planar graph), or 
\item
$X$ is fine or ultra-fine granular input and $Y$ is coarse input, or 
\end{enumerate}
Then, the following statements hold: 
\begin{enumerate}[label=$\blacktriangleright$]
\item
\emph{Any} computational hardness result for $X$ \emph{also} implies the same result for $Y$.
\item
\emph{Any} approximation or exact algorithmic result for $Y$ \emph{also} implies the same result for $X$.
\end{enumerate}
In the statements of our theorems or lemmas, we will use the following convention.
$\kappa>1$ will denote the number of districts.
For polygonal maps (\emph{resp}., planar graph maps)
$\cS$ ((\emph{resp}., $G=(V,E)$) will denote the polygon as a collection of all cells
(\emph{resp}., the graph), 
and $\cS_1,\dots,\cS_\kappa\subset \cS$ 
((\emph{resp}., $V_1,\dots,V_\kappa\subset V$) 
will denote an arbitrary valid (not necessarily optimal) solution.
{\bf
Since every state of USA has a valid current districting partition (sometimes subject to litigation), 
we assume that our problem has already at least one valid (but not necessarily optimal)
solution that can be found in polynomial time (thus, for example, for our computational hardness 
results we are required to exhibit a polynomial-time valid solution).}

{\em In the following two sub-sections, we state our two computational complexity results and some relevant
discussions on them, leaving the actual proofs later in Sections~\ref{sec1}--\ref{sec2}}.

\subsection{Rectilinear polygonal course granularity input}

\begin{theorem}[\bf Hardness of seat-vote equation computation]\label{thm1}
Let $\rho>0,\eps\geq 1$ be two arbitrary finite rational numbers, and 
$c>1, \delta>0$ be any two constants arbitrarily close to $1$ and $0$, respectively.
Suppose that we are allowed a 
(reasonably loose) 
additive 
$|\cS|^{c}$-approximate strict partitioning criteria (\emph{\IE}, 
the partitioning satisfies
$
\textstyle
\max_{1\leq j\leq\kappa} \left\{ \eeta(\cS_j) \right\}
\leq
\min_{1\leq j\leq\kappa} \left\{ \eeta(\cS_j) \right\}
+
|\cS|^{c}
$).
%

\medskip
\noindent
\emph{\textbf{(\emph{a}) (Hardness when $\nvc(\mbox{\partya})<\nicefrac{1}{2}$)}}. 
It is $\NP$-hard to compute an $\eps$-approximation of the seat-vote-equation optimization problem.

\medskip
\noindent
\emph{\textbf{(\emph{b}) (Hardness when $\nvc(\mbox{\partya})\geq\nicefrac{1}{2}$)}}.
Let $\kappa=3\,\alpha+r$ for some two integers $\alpha\geq 1$ and $r\in\{-1,0,1\}$.
Then, it is $\NP$-hard to distinguish between the following two cases:
\begin{enumerate}[label=$\triangleright$]
\item
if the seat-vote-equation has an $(\eps_{\mathrm{low}}-\delta)$-approximation where
$
\eps_{\mathrm{low}} \leq 
\left\{
\begin{array}{r l}
2, & \mbox{if $\kappa\in\{2,3\}$} 
\\
\frac{\kappa}{\alpha+1}-1, & \mbox{otherwise}
\end{array}
\right.
$
\item
or, if the seat-vote-equation has an $(\eps_{\mathrm{high}}+\delta)$-approximation where
$\eps_{\mathrm{high}}\geq \kappa-1$.
\end{enumerate}
Moreover, a valid solution that is a $(\kappa-1)$-approximation always exists irrespective of what definition of 
of an approximately strict partitioning criterion is used.
\end{theorem}

\begin{remark}
The hardness result in 
\emph{\textbf{(\emph{b})}} is tight if $\kappa=2$ since a we have 
a $2$-approximation. For $\kappa>2$ there is a factor gap between the two bounds 
that may be worthy of further investigation.
Note that 
$\lim_{\kappa\to\infty}\eps_{\mathrm{low}}=2$.
\end{remark}


Chatterjee \EA~\cite{CDPAS18} showed that the efficiency gap computation does not admit \emph{any} 
non-trivial approximation at all using the \emph{strict} partitioning criterion if the input is given at
rectilinear polygonal course granularity level.
The following theorem shows that the same result holds even if the strict partitioning criteria is
relaxed \emph{arbitrarily}.

\begin{theorem}[\bf Hardness of efficiency gap computation]\label{thm2}
Let $\delta\geq0,\eps\geq 1$ be any two numbers.
Then, it is $\NP$-hard to compute an $\eps$-approximation of 
$\mathsf{Effgap}_{\kappa}(\cS,\cS_1,\dots,\cS_\kappa)$ 
even when 
we are allowed to use 
$\delta$-approximate $\kappa$-equipartition of $\cS$. 
\end{theorem}

\subsection{Arbitrary polygonal fine granularity input}

For this case, it is clearer to present our proofs if we assume that the planar graph format of
our input, \IE, 
our input is planar graph whose nodes 
are the cells, and whose edges connect pairs of cells if they share a portion of the boundary of non-zero measure.

Chatterjee \EA~\cite{CDPAS18} 
left open the complexity of the efficiency gap computation at the fine granularity level of inputs using 
either exact or approximate partitioning criteria.
Here we show that computing the efficiency gap is $\NP$-complete for arbitrary polygonal fine granularity input even 
under \emph{approximately} strict partitioning criteria.

\begin{theorem}[\bf Hardness of efficiency gap computation]\label{thm-effgap-hard}
Computing $\mathsf{Effgap}_{\kappa}(\cS,\cS_1,\dots,\cS_\kappa)$ is $\NP$-complete 
even when 
we are allowed to use 
$\eps$-approximate $\kappa$-equipartition of $\cS$ 
for any constant $0<\eps<\nicefrac{1}{2}$.
\end{theorem}

\begin{remark}\label{rem2}
The $\NP$-hardness reduction in Theorem~\ref{thm-effgap-hard} does not provide any non-trivial 
inapproximability ratio. In fact, for the specific hard instances of the gerrymandering problem 
constructed in the proof of Theorem~\ref{thm-effgap-hard}, it is possible to design a 
polynomial-time approximation scheme (PTAS) for the efficiency gap computation using 
the approach in~{\em\cite{B94}} $(${\bf the proof of such a PTAS is relatively straightforward and therefore 
we do not provide an explicit proof}$)$. 
\end{remark}

\subsection{What do results and proofs 
in Theorem~\ref{thm1} and Theorem~\ref{thm-effgap-hard} imply in the context of gerrymandering in US?}

Our results are computational hardness result, so one obvious question is about the implications 
of these results and associated proofs for gerrymandering in US. 
To this effect, 
we offer the following motivations and insights that might be of independent interest.
\begin{description}[leftmargin=0.1in]
\item[On following the seat-vote equation:]
Theorem~\ref{thm1} indicates that efficient computation of even a modest approximation to the seat-vote equation 
may be difficult. Thus, unless further research works indicate otherwise, 
it may \emph{not} be a good idea to closely follow the seat-vote equation for computationally 
efficient elimination of gerrymandering (fortunately, many courts also do not recommend on following the seat-vote proportion 
too closely, though not for computational complexity reasons).
\item[On relaxing the exact equipartition criteria:]
Relaxing the exact equipartition criteria even beyond the 
\\
$\thicksim\!\!10$\% margin that has traditionally been 
allowed by courts does \emph{not} seem to make removal of gerrymandering computationally any easier.
\item[On accurate census data at the fine granularity level:]
Accurate census data at the fine granularity level may make a difference to an independent commission 
seeking fair districts (such as in California). As stated in Remark~\ref{rem2},  
while it is difficult to even approximately optimize the absolute difference of the wasted votes 
at a course granularity level of inputs, the situation at the fine granularity level of inputs
may be not so hopeless.
\item[On cracking and packing, how far one can push?]
It is well-known that cracking and packing may result in large partisan bias. For example, 
based on $2012$ election data for election of the (federal) house of representatives for
the states of Virginia, the Democratic party 
had a normalized vote count of about $52$\% but due to cracking/packing held \emph{only} $4$ of the $11$ house seats~\cite{VA1,VA2}.
This observation, coupled with the knowledge that 
Virginia is one of the \emph{most} gerrymandered states in US
both on the congressional and state levels~\cite{web-compact}, 
leads to the following natural question: {\bf ``could the Virginia lawmakers have disadvantaged the Democratic party more
by even more careful execution of cracking and packing approaches''}?  
As one lawmaker put it quite bluntly, they would have liked to gerrymander more \emph{if only} they could.

\hspace*{0.2in}
We believe a partial answer to this is provided by the proof structures for Theorems~\ref{thm2}~and~\ref{thm-effgap-hard}.
A careful inspection of the proofs of 
Theorems~\ref{thm2}~and~~\ref{thm-effgap-hard}
reveal that they \emph{do} use cracking and packing\footnote{For example, packing is used in the proof of  
Theorem~\ref{thm-effgap-hard} when a node $v_i^3$ with $4\delta$ extra supporters for \partya\ is packed 
in the \emph{same} district with the three nodes $v_{i,p}$, $v_{i,q}$ and $v_{i,r}$ each having $\delta$ extra supporters for
\partyb\ (see \FI{ex3-fig}).}
to create hard instances of the efficiency gap minimization problem 
that are computationally intractable to solve optimally certainly 
at the course granularity input level and \emph{even} at the fine granularity input level\footnote{The proofs
of Theorems~\ref{thm2}~and~\ref{thm-effgap-hard} however do \emph{not} make much use of \emph{hijacking} or \emph{kidnapping}.}.
Perhaps the computational complexity issues \emph{did} save the Democratic party from further electoral disadvantages.
\end{description}

\section{Proof of Theorem~\ref{thm1}}
\label{sec1}

\begin{figure}[htbp]
\centerline{\includegraphics[scale=0.85]{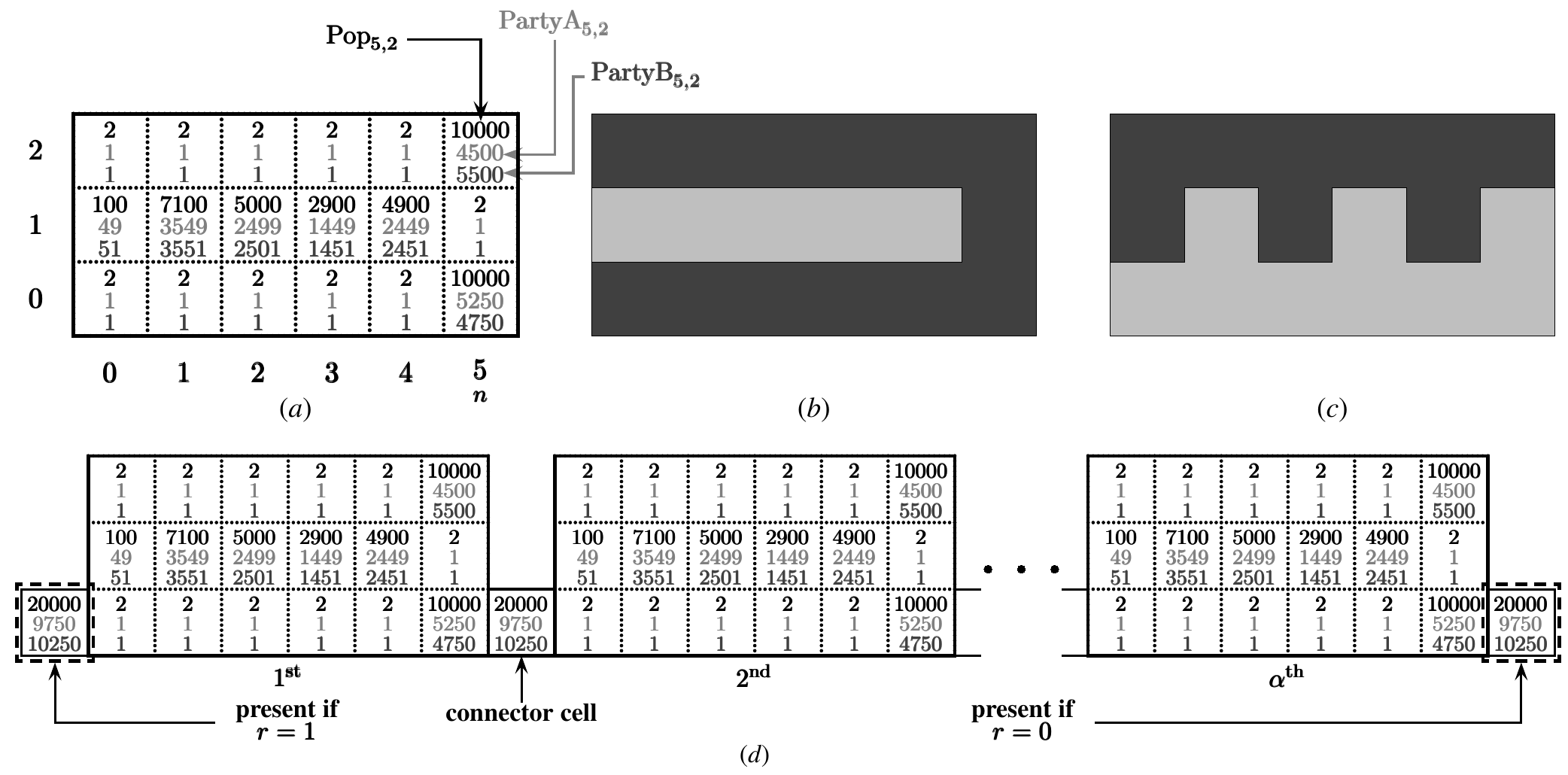}}
%
\caption{\label{ex2-fig}(\emph{a}) An illustration of the construction in the proof of Theorem~\ref{thm1}
for $\kappa=2$ when the instance of the PARTITION problem is 
$\cA=\{100,7100,5000,2900,4900\}$ (and thus $W=20000$).
(\emph{b}) An optimal solution of the redistricting problem when a solution of the PARTITION problem exists.
(\emph{c}) A trivial valid solution which is not optimal.
(\emph{d}) Generalization of the reduction for arbitrary $\kappa\geq 2$.
}
\end{figure}

\noindent
\textbf{(\emph{a})}
We reduce from the $\NP$-complete 
PARTITION 
problem~\cite{GJ79} which is defined as follows: 

\begin{quote}
{\em 
given a set of $n$ positive integers 
$\cA=\left\{ a_0,\dots,a_{n-1}\right\}$, decide if there exists a subset $\cA'\subset\cA$ 
such that 
$
\sum\limits_{a_i\in\cA'} a_i
=
\sum\limits_{a_j\notin\cA'} a_j 
=\frac{W}{2}
$ 
where $W=\sum\limits_{j=0}^{n-1} a_j$ is an even number.
}
\end{quote}

\noindent
Note that we can assume without loss of generality that $n$ is sufficiently large, $n$ and each of 
$a_0,\dots,a_{n-1}$ is a multiple of \emph{any} fixed positive integer (in particular, multiple of $2$), 
$\max_j\{a_j\}<W/2$, \emph{no} two integers in $\cA$ are equal and $W>n^{2c}$.

\medskip
\noindent
\textbf{Proof for $\kappa=2$}.

\smallskip
Multiplying $a_0,\dots,a_{n-1}$ and $W$ by $n^{2+2c}$, and denoting them by the same notations 
we can therefore assume 
that the minimum absolute difference between any two distinct numbers in $\cA$ is at least $n^{2+2c}$
and $W>n^{2+4c}$.
Our rectilinear polygon is a rectangle 
$\cS=\left\{ p_{i,j} \,|\, 0\leq i\leq n, \, 0\leq j\leq 2 \right\}$
of size $3\times (n+1)$ (see~\FI{ex2-fig}(\emph{a})) with the following numbers for various cells:
\begin{gather*}
\eeta_{i,j}=
\left\{
\begin{array}{r l}
a_i, & \mbox{if $0\leq i<n$ and $j=1$}
\\
\nicefrac{W}{2}, & \mbox{if $i=n,j=0$,} 
\\
                 & \,\,\,\, \mbox{or if $i=n,j=2$}
\\
2, & \mbox{otherwise}
\end{array}
\right.
\,\,\,\,\,
\,\,\,\,\,
\aalpha_{i,j}=
\left\{
\begin{array}{r l}
({a_i}/{2})-1, & \mbox{if $0\leq i<n$ and $j=1$}
\\
(\nicefrac{W}{4}) + 50\,n, & \mbox{if $i=n,j=0$} 
\\
(\nicefrac{W}{4}) - 100\,n, & \mbox{if $i=n,j=2$} 
\\
1, & \mbox{otherwise}
\end{array}
\right.
\end{gather*}
Note that:
\begin{enumerate}[label=$\triangleright$]
\item
$\eeta(\cS)=2\times(W/2)+\sum_{j=0}^{n-1}a_j+2\times (2n+1)-2=2W+4n$.
\item
$\aalpha(\cS)=2\times(\nicefrac{W}{4})+50\,n -100\,n
+\sum_{j=0}^{n-1}((a_j/2)-1)+(2n+1)=
W-47n-1$.
\item
$\nvc(\mbox{\partya}) = (W-47n-1)/(2W+4n)<\nicefrac{1}{2}$.
\end{enumerate}
First, as required, we show that $\cS$ has a valid solution satisfying all the constraints. Consider the 
following solution
(refer to \FI{ex2-fig}(\emph{b})):
\[
\cS_1=
\left\{ p_{i,1} \,|\, a_i\in\cA \right\},
\,\,
\,\,
\,\,
\cS_2 = \cC\setminus\cS_1
\]
We can now verify the following:
\[
\eeta(\cS_1)=
\sum_{a_i\in\cA}a_i
= W,\,\,\, 
\eeta(\cS_2)=\eeta(\cS)-\eeta(\cS_1)=W+4n,
\]
and thus the partitioning constraint is satisfied since
$ 4n < (3n+3)^{c}$.
Since $\lim\limits_{\eeta(\cS)\to\infty} \left( \frac{\aalpha(\cS)}{\bbeta(\cS)}\right)=1$, 
the proof is complete once the following claims are shown.

\medskip
\begin{adjustwidth}{0.6cm}{}
\begin{description}
\item[(completeness)]
If the PARTITION problem has a solution 
then 
$\nsc(\mbox{\partya})=1$.
\item[\hspace*{0.2in}(soundness)]
If the PARTITION problem does not have a solution 
then 
$\nsc(\mbox{\partya})=0$.
\end{description}
\end{adjustwidth}

\medskip
\noindent
\textbf{Proof of completeness} 
(refer to \FI{ex2-fig}(\emph{c}))

\smallskip
Suppose that there is a valid solution of $\cA'\subset\cA$ of PARTITION and consider  
the two polygons 
\[
\cS_1=
\left\{ p_{i,0} \,|\, 0\leq i \leq n \right\}
\cup
\left\{ p_{i,1} \,|\, a_i\in\cA' \right\}
\cup
\left\{p_{n,1}\right\},
\,\,
\,\,
\,\,
\cS_2 = \cC\setminus\cS_1
\]
One can now verify the following:
\begin{enumerate}[label=$\triangleright$,leftmargin=*]
\item
$\eeta(\cS_1)=
2 (n+1) + \left(\sum_{a_i\in\cA'}a_i\right) + \frac{W}{2}
=
W + 2n +2$, 
$\eeta(\cS_2)=\eeta(\cS)-\eeta(\cS_1) = W+2n-2$, and thus the partitioning constraint is satisfied since
$\eeta(\cS_1) - \eeta(\cS_2) =
4 < (3n+3)^c$.
\item
$\aalpha(\cS_1)=(n+1)+ \sum_{a_i\in\cA'}\left(\frac{a_i}{2}-1\right) + \frac{W}{4} + 50n
= 
\frac{W}{2} + (51n+ 1) -|\cA'|$, 
$\bbeta(\cS_1)=\eeta(\cS_1)-\aalpha(\cS_1)
=
\frac{W}{2} -49n + 1 +|\cA'|$,
and thus 
$\aalpha(\cS_1)>\bbeta(\cS_1)$
since $|\cA'|<n-1$.
\item
$\aalpha(\cS_2)=\aalpha(\cS)-\aalpha(\cS_1) =\frac{W}{2}- 98n - 2 +|\cA'|$, 
$\bbeta(\cS_2)=\eeta(\cS_2)-\aalpha(\cS_2)=\frac{W}{2}+100n-|\cA'|$,
and thus 
$\aalpha(\cS_2)<\bbeta(\cS_2)$
since $|\cA'|<n-1$.
\end{enumerate}

\medskip
\noindent
\textbf{Proof of soundness} 

\smallskip
Let $\cS_1$ and $\cS_2=\cS \setminus \cS_1$ be the two partitions in any valid solution of the redistricting problem. 
For convenience, let us define the following sets:
\begin{gather*}
\cS_{\cS_1} = \left\{ \,p_{i,1} \,|\,0\leq i < n  \right\} \,\cap\, \cS_1,
\,\,\,\,
\cS_{\cS_2} = \left\{ \,p_{i,1} \,|\,0\leq i < n  \right\} \,\cap\, \cS_2
\\
\cS_{heavy} = \left\{ p_{n,0},p_{n,2} \right\},
\,\,\,\,
\cS_{light} = \cS \setminus \left( \cS_{heavy} \,\cup\, \cS_{\cS_1} \,\cup\, \cS_{\cS_2} \right)
\end{gather*}
The following chain of arguments prove the desired claim.
\begin{description}[leftmargin=0.5cm]
\item[(\emph{i})]
\emph{Both} the cells in 
$\cS_{heavy}$
\emph{cannot} be together in the same partition, say $\cS_1$, with any cell, say $p_{i,1}$, from 
$\cS_{\cS_1} \cup \cS_{\cS_2}$ 
since in that case 
\begin{multline*}
\eeta(\cS_1) \geq W + a_i \,\,\,\&\,\,\, 
\eeta(\cS_2) =  \eeta(\cS)-\eeta(\cS_1) \leq W+4n-a_i 
\\
\Rightarrow\,
\eeta(\cS_1)-\eeta(\cS_2) \geq  2a_i -4n > 2\,n^{2+2c} -4n > n^{2+2c} > |\cS|^c = (3n+3)^c 
\end{multline*}
\item[(\emph{ii})]
At least one of 
$\cS_{\cS_1}$ and $\cS_{\cS_2}$ must be \emph{empty}.  
To see this, assume that both are non-empty. 
By~(\emph{i}), 
we may suppose that $p_{n,0}\in\cS_1$ and $p_{n,2}\in\cS_1$. 
Since the PARTITION problem does not have a solution,  
$
L=\sum\limits_{p_{i,1}\in\cS_{\cS_1}} a_i 
\neq
M=\sum\limits_{p_{i,1}\in\cS_{\cS_2}} a_i 
$.
Assume, without loss of generality, that $L>M$. Then, 
$
L-M \geq \min\limits_{0\leq i<n} \{a_i\}\geq n^{2+2c} 
$, and therefore  
$
|\,\eeta(\cS_1)-\eeta(\cS_2) \,| \geq 
|\, (L-M) - \eeta(\cS_{light}) \,|
>
n^{1+2c}>|\cS|^c
$, 
thus violating the partitioning constraints.
\item[(\emph{iii})]
Since both $\cS_{\cS_1}$ and $\cS_{\cS_2}$ cannot be empty, 
by (\emph{ii}) assume that $\cS_{\cS_1}=\emptyset$ but $\cS_{\cS_2}\neq\emptyset$. Then, 
by (\emph{i}), 
both $p_{n,0}$ and $p_{n,2}$ are in $\cS_1$.
We can now verify that $\nsc(\mbox{\partya})=0$ as follows: 
\begin{itemize}
\item
$
\aalpha(\cS_1) \leq \frac{W}{4}+50n +\frac{W}{4}-100n + 2n+1 = \frac{W}{2} -48n+1
$, 
$
\bbeta(\cS_2) \geq \frac{W}{4}-50n +\frac{W}{4}+100n = \frac{W}{2}+50n
$, 
and thus 
$\aalpha(\cS_1)<\bbeta(\cS_1)$.
\item
$
\aalpha(\cS_2) \leq \sum_{j=0}^{n-1} \left(\frac{a_j}{2}-1\right) +n  = \frac{W}{4}
$, 
$
\bbeta(\cS_2) \geq \sum_{j=0}^{n-1} \left(\frac{a_j}{2}+1\right) \frac{W}{4} +n 
$, 
and thus 
$\aalpha(\cS_1)<\bbeta(\cS_1)$.
\end{itemize}
\end{description}

\medskip
\noindent
\textbf{Proof for $\kappa \geq 2$}.

\smallskip
Let $\kappa=3\,\alpha+r$ for some two integers $\alpha\geq 1$ and $r\in\{-1,0,1\}$.
For this case, we will use $\alpha$
copies, say 
$
\cS^{(1)}, \cS^{(2)},\dots, \cS^{(\alpha)}
$, 
of the $3\times (n+1)$ rectangle $\cS$ used for the previous case 
connected via $\alpha-1$ connector cells, 
say 
$
\cC^{(1)}, \cC^{(2)},\dots, \cC^{(\alpha-1)}
$, 
plus additional one or two cells, say  
$\cC^{(\alpha)}$ and $\cC^{(\alpha+1)}$, 
depending on whether the value of $r$ is $0$ or $1$, respectively 
(refer to \FI{ex2-fig}~(\emph{d})).
We now multiply $a_0,\dots,a_{n-1}$ and $W$ by $n^{3+2c}\kappa^2$, and again denoting them by the same notations 
we can therefore assume 
that the minimum absolute difference between any two distinct numbers in $\cA$ is at least $n^{3+2c}\kappa^2$
and $W>n^{3+4c}\kappa^2$. 
We assign the required numbers to the connector and additional cells as follows: 
$\eeta(\cC^{(j)})=W$ and 
$\aalpha(\cC^{(j)})=\frac{W}{2}-50n$ 
for all $j$.
Letting $\beta=\alpha + r$ denote the actual number of connector cells, we 
now have the following updated calculations: 
\begin{gather*}
\eeta(\cS)=\alpha (2W+4n) + \beta W,
\,\,\,
\aalpha(\cS)=\alpha (W-47n-1) + \beta \left(\frac{W}{2}-50n\right)
\\
|\cS|=3\,\alpha\,(n+1) + (\alpha+r) =
3\,\alpha\,n+ 4\,\alpha + r 
\leq
(\kappa+1)n + \frac{4\kappa}{3} + \frac{7}{3}
<
2\,\kappa\,n
\\
\nvc(\mbox{\partya}) = \frac{\aalpha(\cS)}{\eeta(\cS)} < \nicefrac{1}{2},
\,\,\,
\text{as required}
\end{gather*}

\begin{claim}\label{claim1}
Any of the connector or additional cells cannot appear in the same partition with a cell from 
$\cS^{(j)}_{heavy} = \left\{ p_{jn+(j-1)+r,0},p_{jn+(j-1)+r,2} \right\}$ for any $j$.
\end{claim}

\begin{proof}
Suppose that the connector cell 
$\cC^{(i)}$
is together with at least one of two cells from  
$\cS^{(j)}_{heavy}$
in a partition, say $\cS_p$. 
Then, 
$\eeta(\cS_p) = \frac{W}{2} + W = \frac{3W}{2}
$. 
Note that  
\begin{multline*}
\frac{\eeta(\cS)}{\kappa}
=
\frac{
\alpha (2W+4n) + \beta W
}
{\kappa}
=
\frac{
(2\alpha+\beta)W + 4\alpha n
}
{\kappa}
=
\frac{
(3\alpha+r)W + 4\alpha n
}
{\kappa}
\\
=
\frac{
\kappa W + 4\alpha n
}
{\kappa}
W + 
\frac{4n}{3} \times 
\frac{
\kappa-r
}
{\kappa}
<
W + 
\frac{5n}{3}
\end{multline*}
and thus there exists a partition $\cS_q$, $q\neq p$, such that 
$\eeta(\cS_q) <
W + \frac{5n}{3}$.
Consequently, it follows that 
\[
\eeta(\cS_p)-\eeta(\cS_q)
>
\frac{3W}{2} - 
W + \frac{5n}{3}
>\frac{4W}{3}>\frac{4}{3}n^{3+4c}\kappa^2>|\cS|^c
\]
which violates the partitioning constraint.
\end{proof}

It \emph{is} possible to generalize the proof for $\kappa=2$ to $\kappa>2$. 
Intuitively, if there is a solution to the PARTITION problem then one of the two seats 
in each copy $\cS^{(j)}$ is won by \partya\ but otherwise \partya\ wins no seat at all. 
The correspondingly modified completeness and soundness claims are as follows:

\medskip
\begin{adjustwidth}{0.2cm}{}
\begin{description}
\item[(completeness for $\kappa>2$)]
If the PARTITION problem has a solution 
then 
$\nsc(\mbox{\partya})=\alpha$.
\item[(soundness for $\kappa>2$)]
If the PARTITION problem does not have a solution 
then 
$\nsc(\mbox{\partya})=0$.
\end{description}
\end{adjustwidth}

\medskip
\noindent
\textbf{(\emph{b})}
We can use a proof similar to that in 
\textbf{(\emph{a})}
for $\kappa\geq 2$, but we need to \emph{change} some of the numbers. 
More precisely, 
the cell 
$p_{n+r,0}\in \cS^{(1)}_{heavy}$
in the very first copy 
$\cS^{(1)}$ has the following new number (instead of the previous value of $(W/4)-100\,n$) 
corresponding to the total number of voters for \partya: 
$
\aalpha(p_{n+r,0}) = (W/4) + q\,\alpha^2 n^2
$ where $q\geq 0$ is the smallest integer such that 
$q\alpha^2 n^2 + 100n -49\alpha n-\alpha -50 n \beta \geq 0$.
Note that 
$\bbeta(p_{n+r,0}) = (W/2)-\aalpha(p_{n+r,0})>0$
since 
$W>n^{3+4c}\kappa^2$. 
A relevant calculation is: 
\begin{multline*}
\aalpha(\cS) - \frac{\eeta(\cS)}{2}
=
\left[ \alpha (W-47n-1) + 100n + q\alpha^2 n^2 + \beta \left( \frac{W}{2}-50n \right) \right]
-
\left[\alpha (W+2n) + \beta \frac{W}{2} \right]
\\
=
q\alpha^2 n^2 + 100n -49\alpha n-\alpha -50 n \beta \geq 0
\end{multline*}
and therefore 
$\nvc(\mbox{\partya}) = \frac{\aalpha(\cS)}{\eeta(\cS)} \geq \nicefrac{1}{2}$,
as required.
The only difference in the proofs come from the fact that now in the first copy 
$\cS^{(1)}$ \partya\ always wine one seat by default but wins two seats if PARTITION has a solution. 
The correspondingly modified completeness and soundness claims are as follows:

\medskip
\begin{adjustwidth}{0.6cm}{}
\begin{description}
\item[(modified completeness claim for $\nvc(\mbox{\partya}) > \nicefrac{1}{2}$)]
If the PARTITION problem has a solution 
then 
$\nsc(\mbox{\partya})=\alpha+1$.
\item[(modified soundness claim for $\nvc(\mbox{\partya}) > \nicefrac{1}{2}$)]
If the PARTITION problem does not have a solution 
then 
$\nsc(\mbox{\partya})=1$.
\end{description}
\end{adjustwidth}

\noindent
To see that these completeness and soundness claims indeed prove the desired bounds, note the following: 
\begin{enumerate}[label=$\triangleright$]
\item
$a=\lim\limits_{\eeta(\cS)\to\infty} \left( \frac{\aalpha(\cS)}{\bbeta(\cS)}\right)=1$
and 
$b=\lim\limits_{\eeta(\cS)\to\infty} \left( \frac{\aalpha(\cS)}{\eeta(\cS)}\right)=\nicefrac{1}{2}$. 
\item
If $\kappa=2$ and 
$\nsc(\mbox{\partya})=\alpha+1=2$, 
then $\nsc(\mbox{\partyb})=0$, 
and thus this gives a $2$-approximation since $1/b=2$.
\item
If $\kappa=3$ and 
$\nsc(\mbox{\partya})=\alpha+1=2$, 
then $\nsc(\mbox{\partyb})=1$, 
and thus this gives a $2$-approximation since $\frac{\nsc(\mbox{\partya})}{\nsc(\mbox{\partyb})}=2$.
\item
For any $\kappa\geq 2$, if 
$\nsc(\mbox{\partya})=1$ then 
$\nsc(\mbox{\partyb})=\kappa-1$ 
and thus this gives a $(\kappa-1)$-approximation.
\end{enumerate}
For the existence of a $\kappa$-approximation when 
$\nvc(\mbox{\partya}) \geq \nicefrac{1}{2}$,
note that for any valid solution $\cS_1,\dots,\cS_\kappa$ for $\cS$, 
$\nvc(\mbox{\partya}) = \dfrac{ \aalpha(\cS) } { \eeta(\cS)  }
= \dfrac { \sum_{i=1}^\kappa \aalpha(\cS_i) } { \sum_{i=1}^\kappa \eeta(\cS_i) } 
\geq \nicefrac{1}{2}$, and 
thus there must exists a district $\cS_j$ such that $\aalpha(\cS_j) \geq \nicefrac{\eeta(\cS_j)}{2}$.


\section{Proof sketch of Theorem~\ref{thm2}}

The proof is obtained by carefully modifying the proof of Theorem~4 in~\cite{CDPAS18}
in the following manner:
\begin{enumerate}[label=$\triangleright$,leftmargin=*]
\item
We remove all cells with zero population. As a result, the rectangle in~\cite{CDPAS18}
now becomes a rectilinear polygon (without holes).
\item
We multiply all the non-zero values of $\eeta(\cdot)$'s and $\aalpha(\cdot)$'s by $1+2\delta$.
It is possible to verify that as a result the following claim holds:
\begin{quote}
for any two districts $\cS_i$ and $\cS_j$, 
$\eeta(\cS_i)\neq \eeta(\cS_j)$ implies either 
$\eeta(\cS_i)>(1+\eps)\eeta(\cS_j)$ 
or
$\eeta(\cS_j)>(1+\eps)\eeta(\cS_i)$.  
\end{quote}
This ensures that 
$\eeta(\cS_1)=\dots=\eeta(\cS_\kappa)$
for any valid partition of the rectilinear polygon.
\item
The new soundness and completeness claims now become as follows:
\medskip
\begin{adjustwidth}{0.6cm}{}
\begin{description}
\item[(soundness)]
If the PARTITION problem does not have a solution 
then 
$\mathsf{Effgap}_{\kappa}(\cS,\cS_1,\dots,\cS_\kappa)=\delta\Delta$.
\item[(completeness)]
If the PARTITION problem has a solution 
then 
$\mathsf{Effgap}_{\kappa}(\cS,\cS_1,\dots,\cS_\kappa)=0$.
\end{description}
\end{adjustwidth}
\noindent
where $\Delta$ is exactly as defined in~\cite{CDPAS18}
\end{enumerate}


\section{Proof of Theorem~\ref{thm-effgap-hard}}
\label{sec2}

The problem is trivially in $\NP$, so will concentrate on the $\NP$-hardness reduction.
Our reduction is from the \emph{maximum independent set problem for planar cubic graphs}
(\mispc) 
which is defined as follows: 

\begin{quote}
{\em 
``given a cubic (\emph{\IE}, $3$-regular) planar graph $G=(V,E)$ 
and an integer $\nu$, does there exist an independent set for $G$ with $\nu$ nodes ?''
} 
\end{quote}

\noindent
\mispc\ is known to be $\NP$-complete~\cite{GJS76} but there exists a PTAS for it~\cite{B94}.
Note the value of 
\\
$\mathsf{Effgap}_{\kappa}(\cS,\cS_1,\dots,\cS_\kappa)$ remains the \emph{same} if we divide (or multiply) 
the values of all $\aalpha(\cS_j)$'s and $\bbeta(\cS_j)$'s by $t$ for any integer $t>0$.
Thus, to simplify notation, we assume that we have re-scaled the numbers such that 
$\min_{1\leq j\leq\kappa} \left\{ \eeta(\cS_j) \right\}=1$ and therefore our approximately strict 
partitioning criteria is satisfied by ensuring that 
$1\leq \eeta(\cS_j)\leq 1+\eps$ for all $j=1,\dots,\kappa$ with 
$\eeta(\cS_j)=1$ for at least one $j$.
Thus, each $\aalpha(\cS_j)$, $\bbeta(\cS_j)$ and $\eeta(\cS_j)$
may be positive \emph{rational} constant numbers such that, 
if needed, we can ensure that all these numbers are integers at the end of the reduction by multiplying 
them by a suitable positive integer of polynomial size.

\begin{figure}[htbp]
\centerline{\includegraphics[scale=0.8]{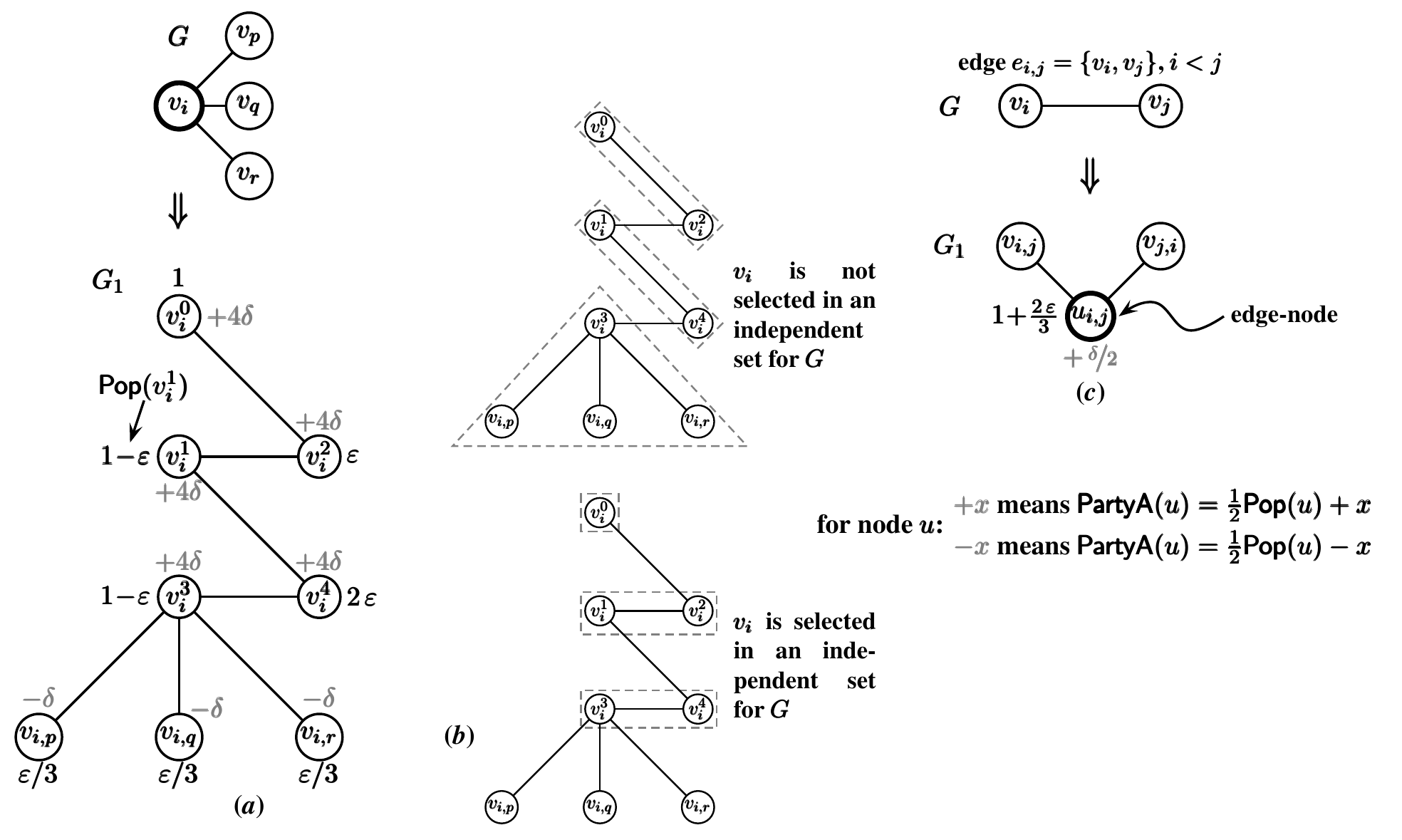}}
%
\caption{\label{ex3-fig}The sub-graph gadgets used in the proof of Theorem~\ref{thm-effgap-hard}.}
\end{figure}

Let $G=(V,E)$ and $\nu$ be the given instance of 
\mispc\
with $V=\{v_1,\dots,v_n\}$ and $|E|=3n/2$. 
Note that, since $G$ is cubic, we can always greedily find an independent set of at least 
$\nicefrac{n}{4}$ nodes and moreover there does not exist any independent set of more than 
$\nicefrac{n}{2}$ nodes; thus we can assume 
$\nicefrac{n}{4}<\nu\leq\nicefrac{n}{2}$.  
Let $\delta=n^{-3}/100>0$ be a rational number of polynomial size that is sufficiently small compared to $\eps$.
We describe an instance of our map $G_1=(V_1,E_1)$ (a planar graph with all required numbers) constructed from $G$ as follows. 
\begin{description}
\item[Node gadgets:]
Every node $v_i\in V$ with its three adjacent nodes as $v_p,v_q,v_r$ is replaced a sub-graph of $8$ new nodes 
$v_i^0,v_i^1,v_i^2,v_i^3,v_i^4,v_{i,p},v_{i,q},v_{i,r}\in V_1$ and $7$ new edges 
along with their $\eeta(\cdot)$ and $\aalpha(\cdot)$ values 
as shown in \FI{ex3-fig}(\emph{a}).
The requirement ``$1\leq \eeta(\cS_j)\leq 1+\eps$ for all $j$'' and the fact 
that $0<\eps<\nicefrac{1}{2}$ ensure that these nodes can be covered only in 
the two possible ways as shown in \FI{ex3-fig}(\emph{b}):
\begin{enumerate}[label=$\triangleright$,leftmargin=*]
\item
For the top case in \FI{ex3-fig}(\emph{b}),
all the $8$ nodes are covered by $3$ districts.
Intuitively, this corresponds to the case when $v_i$ is \emph{not} selected in an independent set for $G$.
We informally refer to this as the \emph{the ``$v_i$ is not selected'' case}.
\item
For the bottom case in \FI{ex3-fig}(\emph{b}),
$5$ of the $8$ nodes are covered by $3$ districts, leaving the remaining $3$ nodes 
(nodes $v_{i,p},v_{i,q},v_{i,r}$) to be covered with some other nodes in $G_1$. 
Intuitively, this corresponds to the case when $v_i$ \emph{is} selected in an independent set for $G$.
We informally refer to this as the \emph{the ``$v_i$ is selected'' case}.
\end{enumerate}
Note that this step in all introduces $8n$ new nodes and $7n$ new edges in $G_1$. 
\item[Edge gadgets:]
For every edge $e_{i,j}=\{v_i,v_j\}\in E$ (with $i<j$), 
we introduce one new node (the ``edge-node'') 
$u_{i,j}$ and two new edges 
$\{v_{i,j},u_{i,j}\}$ and 
$\{v_{j,i},u_{i,j}\}$
as shown in \FI{ex3-fig}(\emph{c}).
Note that this step in all introduces $3n/2$ new nodes and $3n$ new edges in $G_1$. 
\end{description}
Thus, we have $|V_1|=19n/2$ and $|E_1|=10n$, and surely $G_1$ is planar since $G$ was 
a planar graph. Finally, we set $\kappa=9n/2$. \emph{Note that 
the instance $G_1$ is at the fine granularity level since the total population 
of every node is between $\eps/3$ and $1+(2\eps/3)$ for a constant $\eps$}.

To continue with the proof, we need to make a sequence of observations about the constructed
graph $G_1$ as follows: 
\begin{description}
\item[(\emph{i})]
An edge-node $u_{i,j}$ can be in a partition just by itself, or with only one of either of the nodes $v_{i,j}$ and $v_{j,i}$.
\item[(\emph{ii})]
If $v_i$ is not selected then $u_{i,j}$ cannot be in the same partition as $v_{i,j}$.
On the other hand, if 
$u_{i,j}$ is in the same partition as $v_{i,j}$ then $v_i$ must be selected.
\item[(\emph{iii})]
By \textbf{(\emph{i})} and \textbf{(\emph{ii})}, 
An edge-node $u_{i,j}$ is in a partition just by itself if and only if 
neither of its end-points, namely nodes $v_i$ and $v_j$, are selected in the corresponding independent set for $G$.
\item[(\emph{iv})]
Consider any maximal independent set $\emptyset\subset V'\subset V$ for $G$ (\EG, the one obtained by the 
obvious greedy solution) having $0<\mu<\nicefrac{n}{2}$ nodes. 
Using \textbf{(\emph{i})}, \textbf{(\emph{ii})} and 
\textbf{(\emph{iii})}, 
the following calculations hold:  
\begin{enumerate}[label=$\triangleright$]
\item
For every node $v_i$ selected in $V'$ with its adjacent nodes being $v_p,v_q,v_r$, we cover
the nodes 
$v_i^0$, $v_i^1$, $v_i^2$, $v_i^3$, $v_i^4$, $v_{i,p}$, $v_{i,q}$, $v_{i,r}$, and 
the three edge-nodes corresponding to the three edges 
$\{v_i,v_p\}$, $\{v_i,v_q\}$, $\{v_i,v_r\}\in E$
using $6$ districts in $G_1$.
\item
For every node $v_i$ \emph{not} selected in $V'$,  
we cover the nodes 
$v_i^0$, $v_i^1$, $v_i^2$, $v_i^3$, $v_i^4$, $v_{i,p}$, $v_{i,q}$, and $v_{i,r}$
using $3$ districts in $G_1$.
\item
Let $E'\subseteq E$ be the set of edges such that \emph{neither} end-points of these edges are
selected in $V'$. Note that $|E'|=(3n/2)-3\mu$, and for every edge $v_{i,j}\in E'$ 
we use one new district for the edge-node $u_{i,j}$.
\end{enumerate}
\end{description}

\begin{lemma}[existence of valid solution]
There is a trivial (not necessarily optimal) valid solution for $G_1$.
\end{lemma}

\begin{proof}
By \textbf{(\emph{iv})},
the total number of districts used in a maximal independent set is
$6\mu + 3 (n-\mu) + ( (3n/2) - 3\mu) = 9n/2 = \kappa$, 
as required.
\end{proof}

Next, for calculations of the wasted votes and the corresponding efficiency gap,
we remind the reader of the following calculations for a district $\cS_j$ (for any sufficiently
small positive rational number $x$): 
\begin{gather*}
\wv(\cS_j,\mbox{\partya}) =
\left\{
\begin{array}{r l}
   x, & \mbox{if $\aalpha(\cS_j)=\frac{\eeta(\cS_j)}{2} + x$}
\\
[3pt]
\frac{\eeta(\cS_j)}{2} - x, & \mbox{if $\aalpha(\cS_j)=\frac{\eeta(\cS_j)}{2} - x$}
\end{array}
\right.
\\
\wv(\cS_j,\mbox{\partyb}) =
\left\{
\begin{array}{r l}
\frac{\eeta(\cS_j)}{2} - x, & \mbox{if $\aalpha(\cS_j)=\frac{\eeta(\cS_j}{2} + x$}
\\
[3pt]
x, & \mbox{if $\aalpha(\cS_j)=\frac{\eeta(\cS_j)}{2} - x$}
\end{array}
\right.
\end{gather*}
\begin{multline*}
\wv(\cS_j,\mbox{\partya}) - \wv(\cS_j,\mbox{\partyb}) 
\\
=
\left\{
\begin{array}{r l}
2x - \frac{\eeta(\cS_j)}{2}, & \mbox{if $\aalpha(\cS_j)=\frac{\eeta(\cS_j)}{2} + x$}
\\
[3pt]
\frac{\eeta(\cS_j)}{2} - 2x, & \mbox{if $\aalpha(\cS_j)=\frac{\eeta(\cS_j)}{2} - x$}
\end{array}
\right.
\end{multline*}
Consider any maximal independent set $\emptyset\subset V'\subset V$ for $G$ 
having $\nicefrac{n}{4}<\mu\leq\nicefrac{n}{2}$ nodes. 
Using \textbf{(\emph{iv})},
the following calculations hold:  
\begin{enumerate}[label=$\triangleright$]
\item
Every node $v_i$ selected in $V'$ 
contributes the following amount 
to the total value of 
\\
$\sum\limits_{j=1}^{\kappa} (\wv(\cS_j,\mbox{\partya}) - \wv(\cS_j,\mbox{\partyb}) )$:
\begin{gather*}
\xi=
\left( 8\delta - \frac{1}{2} \right) + 
\left( 16\delta - \frac{1}{2} \right) + 
\left( 16\delta - \frac{1+\eps}{2} \right) + 
3 \times \left( \frac{1}{2} - 2\delta \right)
=
34\delta - \frac{\eps}{2}
\end{gather*}
\item
Every node $v_i$ \emph{not} selected in $V'$ 
contributes the following amount 
to the total value of 
\\
$\sum\limits_{j=1}^{\kappa} (\wv(\cS_j,\mbox{\partya}) - \wv(\cS_j,\mbox{\partyb}) )$:
\begin{gather*}
\zeta=
\left( 16\delta - \frac{1+\eps}{2} \right) + 
\left( 16\delta - \frac{1+\eps}{2} \right) + 
\left( 2\delta - \frac{1}{2} \right)
=
34\delta -\eps - \frac{3}{2}
\end{gather*}
\item
Every edge in $E$ such that \emph{neither} end-points of the edge are
selected in $V'$
contributes the following amount 
to the total value of 
$\sum\limits_{j=1}^{\kappa} (\wv(\cS_j,\mbox{\partya}) - \wv(\cS_j,\mbox{\partyb}) )$:
\begin{gather*}
\eta=
\delta
-
\frac{1+\frac{2\eps}{3}}{2} 
=
\delta  - \frac{\eps}{3} - \frac{1}{2}
\end{gather*}
\item
Consequently, adding all the contributions, 
we get the following value for 
$\sum\limits_{j=1}^{\kappa} (\wv(\cS_j,\mbox{\partya}) - \wv(\cS_j,\mbox{\partyb}) )$
corresponding to an independent set of $\mu$ nodes: 
\begin{multline*}
\Upsilon(\mu) 
= 
\mu\xi + (n-\mu)\zeta + \left( \frac{3n}{2} - 3\mu \right) \eta  
\\
=
\left( 34\mu\delta - \frac{\mu\eps}{2} \right) + 
(n-\mu) \left( 34\delta -\eps - \frac{3}{2} \right) +
\left( \frac{3n}{2} - 3\mu \right) \left( \delta - \frac{\eps}{3}  - \frac{1}{2} \right) 
\\
=
3\mu + 
\left( \frac{3\eps}{2} - 3\delta \right) \mu 
+
\left( \frac{71\delta}{2} - \frac{3\eps}{2}  - \frac{9}{4} \right) n
\end{multline*}
\end{enumerate}
Now we note the following properties of the quantity $\Upsilon(\mu)$:
\begin{enumerate}[label=$\triangleright$]
\item
Since $\delta = n^{-3}/100$ and 
$\nicefrac{n}{4}<\mu\leq\nicefrac{n}{2}$, we have 
$\Upsilon(\mu) < 0$ and therefore $|\Upsilon(\mu)| = -\Upsilon(\mu)$.
\item
Consequently, 
$
|\Upsilon(\mu)| - |\Upsilon(\mu-1)|
=
\Upsilon(\mu-1) - \Upsilon(\mu)
= -3 - \frac{3\eps}{2} + 3\delta 
$
\end{enumerate}
The last equality then leads to the following two statements that complete the proof for $\NP$-hardness:
\begin{enumerate}[label=$\blacktriangleright$,leftmargin=*]
\item
If $G$ has an independent set of $\nu$ nodes then 
$\mathsf{Effgap}_{\kappa}(\cS,\cS_1,\dots,\cS_\kappa)=|\Upsilon(\nu)|$.
\item
If every independent set of $G$ has at most $\nu-1$ nodes then 
$\mathsf{Effgap}_{\kappa}(\cS,\cS_1,\dots,\cS_\kappa)\geq|\Upsilon(\nu-1)|>|\Upsilon(\nu)|+2$.
\end{enumerate}

\section{Concluding remarks}

The computational complexity results in this article (and also in~\cite{CDPAS18}) 
may be considered as a beginning to gerrymandering from a \emph{TCS} point of view. 
While some computational complexity aspects of these problems are settled, a plethora
of interesting \emph{TCS}-related questions remaining. Some of these questions are as follows.
\begin{enumerate}[label=$\triangleright$]
\item
The computational complexity of optimizing the partisan bias measure remains wide open. 
Of special interest is the \emph{uniform population shift} model for which 
$\beta_1=\dots=\beta_\kappa = \nicefrac{\alpha}{\kappa}$.
\item
Does introducing the additional constraint of geometric compactness render the computation 
of the gerrymandering objectives more tractable? Theorem $11$ of~\cite{CDPAS18}
provides a partial (affirmative) answer to this question for restricted versions of efficiency 
gap calculation problem.
\item
Is there a constant factor approximation algorithm for computing the efficient gap measure for 
inputs at a fine granularity level? We conjecture this to be true but have been unable to prove it yet.
\end{enumerate}

\section*{Acknowledgments}

We thank Laura Palmieri and Anastasios Sidiropoulos for useful discussions.




\begin{thebibliography}{99}
\bibitem{Alt02}
M. Altman, \emph{A Bayesian approach to detecting electoral manipulation},
Political Geography, 21, 39-48, 2002.
\bibitem{B94}
B. S. Baker, \emph{Approximation algorithms for NP-complete problems on planar graphs},
Journal of the Association for Computing Machinery, 41(1), 153-180, 1994.
\bibitem{amici04}
\emph{Brief of amici curiae in support of neither party}, Vieth v.\ Jubelirer, 541 U.S. 267, 2004.
\bibitem{BK87}
R. X. Browning and G. King,
\emph{Seats, Votes, and Gerrymandering: Estimating Representation and Bias in State Legislative Redistricting},
Law \& Policy, 9(3), 305-322, 1987. 
\bibitem{C85}
E. B. Cain, \emph{Simple v. complex criteria for partisan gerrymandering: a comment on Niemi and Grofman},
UCLA Law Review, 33, 213-226, 1985.
\bibitem{CDPAS18}
T. Chatterjee, B. DasGupta, L. Palmieri, Z. Al-Qurashi and A. Sidiropoulos, 
\emph{Alleviating partisan gerrymandering: can math and computers help to eliminate wasted votes?}, 
arXiv:1804.10577, 2018.
\bibitem{CR15}
J. Chen and J. Rodden, \emph{Cutting through the thicket: redistricting simulations and the detection of partisan gerrymanders},
Election Law Journal, 14(4), 331-345, 2015.
\bibitem{wcl16}
W. K. T. Cho and Y. Y. Liu,
\emph{Toward a talismanic redistricting tool: a computational method for identifying extreme redistricting plans},
Election Law Journal: Rules, Politics, and Policy, 15(4), 351-366, 2016.
\bibitem{CDO00}
C. Cirincione, T. A. Darling and T. G. O'Rourke, \emph{Assessing South Carolina's 1990s Congressional Redistricting},
Political Geography, 19, 189-211, 2000.
\bibitem{t1}
Davis v.\ Bandemer, 478 US 109, 1986.
\bibitem{D15}
S. Doyle, \emph{A Graph Partitioning Model of Congressional Redistricting},
Rose-Hulman Undergraduate Mathematics Journal, 16(2) 38-52, 2015.
\bibitem{Fa89}
D. L. Faigman, \emph{To have and have not: assessing the value of social science to the law as science and policy},
Emory Law Journal, 38, 1005-1095, 1989.
\bibitem{GJ79} 
M. R. Garey and D. S. Johnson, {\em Computers and Intractability -
A Guide to the Theory of NP-Completeness}, W. H. Freeman \& Co., San Francisco, CA, 1979.
\bibitem{GJS76}
M. R. Garey, D. S. Johnson and L. Stockmeyer, 
\emph{Some simplified NP-complete graph problems},
Theoretical Computer Science, 1, 237-267, 1976.
\bibitem{AG94}
A. Gelman and G. King, \emph{A unified method of evaluating electoral systems and redistricting plans},
American Journal of Political Science, 38(2), 514-554, 1994.
\bibitem{J94}
S. Jackman, \emph{Measuring electoral bias: Australia, 1949-93},
British Journal of Political Science, 24(3), 319-357, 1994.
\bibitem{KS50}
M. G. Kendall and A. Stuart, \emph{The Law of Cubic Proportions in Election
Results}, British Journal of Sociology, 1, 183-197, 1950.
\bibitem{t2}
League of united latin american citizens v.\ Perry, 548 US 399, 2006.
\bibitem{LWCW16}
Y. Y. Liu, K. Wendy, T. Cho and S. Wang,
\emph{PEAR: a massively parallel evolutionary computational approach for political redistricting optimization and analysis},
Swarm and Evolutionary Computation, 30, 78-92, 2016.
\bibitem{MB15}
M. McDonald and R. Best, 
\emph{Unfair partisan gerrymanders in politics and law: a diagnostic applied to six cases},
Election Law Journal, 14(4), 312-330, 2015.
\bibitem{m14}
E. McGhee,
\emph{Measuring partisan bias in single-member district electoral systems},
Legislative Studies Quarterly, 39(1), 55-85, 2014.
\bibitem{ND78}
R. G. Niemi and J. Deegan, \emph{A theory of political districting},
The American Political Science Review, 72(4), 1304-1323, 1978.
\bibitem{NGCH90}
R. G. Niemi, B. Grofman, C. Carlucci and T. Hofeller, 
\emph{Measuring compactness and the role of a compactness standard in a test for partisan and racial gerrymandering},
Journal of Politics, 52(4), 1155-1181, 1990.
\bibitem{O78}
O. Osserman, \emph{Isoperimetric Inequality}, 
Bulletin of the American Mathematical Society, 84(6), 1182-1238, 1978.
\bibitem{PPY17}
W. Pegden, A. D. Procaccia and D. Yu,
\emph{A partisan districting protocol with provably nonpartisan outcomes},
arXiv:1710.08781v1 [cs.GT], 2017.
\bibitem{PLB11}
O. Pierce, J. Larson and L. Beckett, \emph{Redistricting, a Devil's Dictionary}, ProPublica, 2011. 
\bibitem{PP91}
D. Polsby and R. Popper, 
\emph{The Third Criterion: Compactness as a Procedural Safeguard Against Partisan Gerrymandering},
Yale Law and Policy Review, 9(2), 301-353, 1991.
\bibitem{RuCo19}
Rucho \EA v.\ Common Cause \EA, No. 18-422, argued March 26, 2019 --- decided June 27, 2019.
\bibitem{R03}
J. E. Ryan, \emph{The limited influence of social science evidence in modern desegregation cases},
North Carolina Law Review, 81(4), 1659-1702, 2003.
\bibitem{sm15}
N. Stephanopoulos and E. McGhee,
\emph{Partisan gerrymandering and the efficiency gap},
University of Chicago Law Review, 82(2), 831-900, 2015.
\bibitem{TL67}
J. Thoreson and J. Liittschwager,
\emph{Computers in behavioral science: legislative districting by computer simulation},
Behavioral Science, 12, 237-247, 1967.
\bibitem{T73}
R. Taagepera, 
\emph{Seats and Votes: A Generalization of the Cube Law of Elections},
Social Science Research, 2, 257-275, 1973.
\bibitem{noteq}
US Supreme Court ruling in Karcher v.\ Daggett, 1983.
\bibitem{W18}
G. S. Warrington, \emph{Quantifying gerrymandering using the vote distribution},
Election Law Journal: Rules, Politics, and Policy, 17(1), 39-57, 2018.
\bibitem{VA1}
\url{http://elections.nbcnews.com/ns/politics/2012/Virginia}
\bibitem{VA2}
\url{http://www.virginiaplaces.org/government/congdist.html}
\bibitem{wiki}
\url{https://en.wikipedia.org/wiki/Gerrymandering}
\bibitem{web-compact}
\url{https://en.wikipedia.org/wiki/Virginia's_congressional_districts}, see also \url{https://www.onevirginia2021.org/}
\end{thebibliography}
\end{document}